\documentclass[runningheads,a4paper]{llncs}
\usepackage[bookmarks=true]{hyperref}
\usepackage{amssymb}

\usepackage[cmex10]{amsmath}

\usepackage{ntheorem}
\theoremstyle{break}

\usepackage{graphicx}
\usepackage{verbatim}
\usepackage{tikz}
\usetikzlibrary{arrows,positioning,shadows}
\usepackage{pgf,pgfarrows,pgfnodes,pgfautomata,pgfheaps}
\usepackage[ruled,vlined,linesnumbered]{algorithm2e}

\usepackage{authblk}

\usepackage{url}
\urldef{\mailsa}\path|{alfred.hofmann, ursula.barth, ingrid.haas, frank.holzwarth,|
\urldef{\mailsb}\path|anna.kramer, leonie.kunz, christine.reiss, nicole.sator,|
\urldef{\mailsc}\path|erika.siebert-cole, peter.strasser, lncs}@springer.com|
\newcommand{\keywords}[1]{\par\addvspace\baselineskip
\noindent\keywordname\enspace\ignorespaces#1}

\tikzset{
  treenode/.style = {align=center, inner sep=0pt, text centered,
    font=\sffamily},
  arn_n/.style = {treenode, circle, black,  draw=black,
     text width=2em},
  arn_r/.style = {treenode, circle, red, draw=red,
    text width=2em, very thick},
  arn_x/.style = {treenode, rectangle, draw=black,
    minimum width=0.5em, minimum height=0.5em}
}

\newcommand{\alert}{\textcolor{red}}

\begin{document}

\mainmatter  


\title{Inherit Differential Privacy in Distributed Setting:\\ Multiparty Randomized Function Computation}

\author{Genqiang Wu\inst{1,2} \and Yeping He\inst{1} \and Jingzheng Wu\inst{1} \and Xianyao Xia\inst{1}}

\institute{NFS, Institute of Software Chinese Academy of Sciences,
Beijing 100190, China \\
\email{genqiang80@gmail.com, \{yeping,jingzheng,xianyao\}@nfs.iscas.ac.cn}
\and
SIE, Lanzhou University of Finance and Economics, Lanzhou 730020, China
}
\authorrunning{Genqiang Wu et al.}
\titlerunning{Inherit Differential Privacy in Distributed Setting}

\maketitle

\begin{abstract}
How to achieve differential privacy in the distributed setting, where the dataset is distributed among the distrustful parties, is an important problem. We consider in what condition can a protocol inherit the differential privacy property of a function it computes. The heart of the problem is the secure multiparty computation of randomized function. A notion \emph{obliviousness} is introduced, which captures the key security problems when computing a randomized function from a deterministic one in the distributed setting. By this observation, a sufficient and necessary condition about computing a randomized function from a deterministic one is given. The above result can not only be used to determine whether a protocol computing differentially private function is secure, but also be used to construct secure one. Then we prove that the differential privacy property of a function can be inherited by the protocol computing it if the protocol privately computes it. A composition theorem of differentially private protocols is also presented. We also construct some protocols to generate random variate in the distributed setting, such as the uniform random variates and the inversion method. By using these fundamental protocols, we construct protocols of the Gaussian mechanism, the Laplace mechanism and the Exponential mechanism. Importantly, all these protocols satisfy obliviousness and so can be proved to be secure in a simulation based manner. We also provide a complexity bound of computing randomized function in the distribute setting. Finally, to show that our results are fundamental and powerful to multiparty differential privacy, we construct a differentially private empirical risk minimization protocol.
\end{abstract}

\section{Introduction} \label{sec-introduction}

Nowadays, a lot of personal information are collected and stored in many databases. Each database is owned by a particular autonomous entity, e.g., financial data by banks, medical data by hospitals, online shopping data by e-commerce companies, online searching records by search engine companies, income data by tax agencies. Some entities may want to mine useful information among these databases. For example, insurance companies may want to analyze the insurance risk of some group by mining both the bank's database and the hospital's database, or several banks may want to aggregate their databases to estimate the loan risk in some area, or, more generally, one may want to learn a classifier among these private databases \cite{DBLP:conf/nips/PathakRR10}. However, due to privacy consideration, data integrating or data mining among these databases should be conducted in a privacy-preserving way: First, one must perform computations on database that must be kept private and there is no single entity that is allowed to see all the databases on which the analysis is run; Second, it is not a priori clear whether the analysis results contain sensitive information traceable back to particular individuals \cite{DBLP:conf/kdd/GantaKS08,DBLP:conf/sp/NarayananS08}. The first privacy problem is the research field of secure MultiParty Computation (MPC) \cite{DBLP:books/cu/Goldreich2004}. However, since standard MPC does not analyze and prevent what is (implicitly) leaked by the analysis results \cite{DBLP:conf/crypto/BeimelNO08,DBLP:conf/pldi/MardzielHKS12}, the second privacy problem can not be treated by MPC. Fortunately, the second privacy problem could be analyzed by differential privacy (DP) \cite{DBLP:journals/fttcs/DworkR14,DBLP:conf/icalp/Dwork06}, which is a mathematically rigorous privacy model that has recently received a significant amount of research attention for its robustness to known attacks, such as those involving side information \cite{DBLP:conf/kdd/GantaKS08,DBLP:conf/sp/NarayananS08}. Therefore, solving the above privacy problems needs the combination of MPC and DP as a tool.


There is a misunderstanding that the above problem can easily be solved without using MPC: Each party first locally analyzes and perturbs the local data using the appropriate differentially private algorithm and then outputs the result; These results are then synthesized to obtain the final result. Obviously, the final result satisfies differential privacy. However, the above method will either add more noise to the final result, such as in the noise mechanism \cite{DBLP:conf/icalp/Dwork06}, or need redesign of the related algorithm, such as in the exponential mechanism \cite{DBLP:conf/focs/McSherryT07}, which would be a more hard work.



We now present the considered problem in a more formal way. Let a dataset $x=(x_1, \ldots, x_n)$ be distributed among the mutually distrustful parties $P_1, \ldots, P_n$, where $x_i$ is owned by $P_i$. We call the above dataset owning condition by \emph{the distributed setting}. The parties want to implement differentially private analyses in the distributed setting by the following way: First choose what to compute, i.e., a differentially private function $\mathcal M(x)$; Then decide how to compute it, i.e., construct an MPC protocol to compute $\mathcal M(x)$. In the paper we only treat the second step. That is, we assume that there has been a differentially private algorithm $\mathcal M(x)$ in the client-server setting. Our task is to construct an MPC protocol $\pi$ to compute $\mathcal M(x)$ in the distributed setting. Furthermore, it is vital that $\pi$ should `inherit' the differential privacy property of $\mathcal M(x)$. That is, in executing $\pi$, each party's view (or each subgroup of the parties' views) should be differentially private to other parties' private data. However, constructing such protocol is challenging. To see that, we consider two examples appeared in the related works to construct differentially private protocols.

\begin{example}[Gaussian mechanism] \label{example1}
The party $P_i$ has the math score list $x_i$ of Class $i$ for $i= 1,2$. $P_1,P_2$ are willing to count the total number of the students whose score $\ge 60$ in $x_1$ and $x_2$, while letting the score list one party owns to be secret to the other and letting the output $f(x_1,x_2)$ satisfies differential privacy, where $f$ is the counting function. We use Gaussian mechanism to achieve differential privacy, i.e., adding Gaussian noise to $f(x_1,x_2)$. Note that the sensitivity of $f$ is $\Delta f = 1$. Therefore, we can add a random number $N \sim \mathcal N(0, \sigma^2)$ to achieve $(\epsilon,\delta)$-differential privacy, where $\sigma >\sqrt{2 \ln 1.25/\delta}/\epsilon$ \cite[page 471]{DBLP:journals/fttcs/DworkR14}. There are two intuitive protocols to achieve the task:
    \begin{enumerate}
    \item Each $P_i$ generates a random number $N_i \sim \mathcal N(0, \sigma^2/2)$ and computes $o_i = f(x_i)+N_i$ locally. $P_1, P_2$ then compute $o_1 +o_2$ using an MPC protocol and output the result $o$. Note that $o = f(x_1, x_2)+(N_1+N_2)$ since $f(x_1,x_2) = f(x_1) +f(x_2)$ and that $(N_1+N_2)\sim \mathcal N(0, \sigma^2)$ due to the infinitely divisibility of Gaussian noise.
    \item Each $P_i$ generates a random number $N_i' \sim \mathcal N(0, \sigma^2)$ locally. $P_1, P_2$ then compute and output $ o = f(x_1) +f(x_2) + \mathrm{LT}(N_1',N_2')$ using an MPC protocol, where $\mathrm{LT}(N_1',N_2')$ outputs the smaller one in $N_1', N_2'$.
    \end{enumerate}
\end{example}

Intuitively, both of the two protocols in Example \ref{example1} satisfy $(\epsilon,\delta)$-differential privacy since both of them add noises drawn from $\mathcal N(0, \sigma^2)$ to $f(x_1,x_2)$. However, to the first protocol, if $P_1$ computes $o - N_1$ it obtains the vale of $f(x_1, x_2) + N_2$. Since $N_2 \sim \mathcal N(0, \sigma^2/2)$ but not $N_2 \sim \mathcal N(0, \sigma^2)$, $P_1$ obtains an output not satisfying $(\epsilon,\delta)$-differential privacy. To the second protocol, either $N_1' = \mathrm{LT}(N_1',N_2')$ or $N_2' = \mathrm{LT}(N_1',N_2')$. Without loss of generality, assuming $N_1' = \mathrm{LT}(N_1',N_2')$, $P_1$ can then compute the value of $o- N_1'$ to obtain $f(x_1, x_2)$, which obviously violates differential privacy. A similar protocol, which has the similar drawback as the second protocol, is used to generate Laplace noise in the distributed setting in \cite{DBLP:conf/nips/PathakRR10}. Therefore, both of the two protocols in Example \ref{example1} do not inherit the $(\epsilon,\delta)$-differential privacy property of the function they compute.

\begin{example}[Laplace mechanism] \label{example3}
The same as Example \ref{example1}, $P_1, P_2$ want to output $f(x_1,x_2)$. In this time, they use Laplace mechanism to achieve differential privacy, i.e., adding Laplace noise to $f(x_1,x_2)$. Since $\Delta f = 1$, they can add a random number $N \sim \mathrm{Lap}(1/\epsilon)$ to achieve $\epsilon$-differential privacy. They construct a protocol as follows: Each party $P_i$ generates two random numbers $Y_{i1}, Y_{i2}$ drawn from $\mathcal N(0, 1/\sqrt{2\epsilon})$ locally. The parties then use an MPC protocol to compute $o=f(x_1,x_2)+N$ and output $o$, where $N \gets \sum_i(Y_{i1}^2-Y_{i2}^2)$.
\end{example}

The above protocol is shown in \cite{DBLP:conf/sigmod/RastogiN10,DBLP:journals/tdsc/MohammedAFD14,7286780}. However, we conclude that it does not inherit the $\epsilon$-differential privacy property of the function it computes. The reason is that $P_1$ can obtain the value of $f(x_1,x_2)+(Y_{21}^2-Y_{22}^2)$ by subtracting $(Y_{11}^2-Y_{12}^2)$ from $o$. However, since the distribution function of $(Y_{21}^2-Y_{22}^2)$ is not $\mathrm{Lap}(1/\epsilon)$ the value of $f(x_1,x_2)+(Y_{21}^2-Y_{22}^2)$ will not satisfy $\epsilon$-differential privacy. 

From Example \ref{example1} and Example \ref{example3} we see that it is difficult to construct a protocol that can inherit the differential privacy property of the function it computes. The crux of the difficulty is that differentially private function is a kind of randomized function, whose output is a random element drawn from a prescribed distribution function (please see Definition \ref{definition-randomized} in Section \ref{sub-sec-dp}) and that the result about computing randomized function in MPC is rare. In the paper we will develop some theoretical results about computing randomized function in the distributed setting and then treat the above inheritance problem. Note that differentially private function and random variate are two kinds of randomized function: with constant inputs for the second one.

\subsection{Contribution}

Our contributions are as follows.


First, we provide a special security definition of computing randomized function in the distributed setting, in which a new notion \emph{obliviousness} is introduced. Obliviousness captures the key security problems when computing a randomized function from a deterministic one. By this observation, we provide a sufficient and necessary condition (Theorem \ref{theorem-obliviousness-mpc-security}) about computing a randomized function from a deterministic one. The above result can not only be used to determine whether a protocol computing a randomized function (and therefore computing a differentially private function) is secure, but also be used to construct secure one. To the best of our knowledge, ours (Theorem \ref{theorem-obliviousness-mpc-security}) is the first to provide a sufficient and necessary condition about this problem.




Second, we prove that a differentially private algorithm can preserve differential privacy property in the distributed setting if the protocol computing it is secure (Theorem \ref{theorem-mdp}), i.e., the \emph{inheritance} problem. We also introduce the composition theorem of differential privacy in the distributed setting (Theorem \ref{theorem-composition-DP}). To the best of our knowledge, the paper is the first to present these results in differential privacy.

Third, we construct some fundamental protocols to generate random variate in the distributed setting, such as Protocol \ref{protocol-Bernoulli-distributed} and Protocol \ref{protocol-inversion-distributed}. By using these fundamental protocols, we construct protocols of the Gaussian mechanism (Protocol \ref{protocol-Gaussian-mechanism-distributed}), the Laplace mechanism (Protocol \ref{protocol-laplace-mechanism}) and the Exponential mechanism (Protocol \ref{protocol-discrete-exponential-distributed} and Protocol \ref{protocol-gibbs-sampler}). To the best of our knowledge, Protocol \ref{protocol-gibbs-sampler} is the first exponential mechanism to treat high-dimensional continuous range in the distributed setting. Importantly, all these protocols satisfy obliviousness and, therefore, can be proved to be secure in a simulation based manner by using the conclusion of Theorem \ref{theorem-obliviousness-mpc-security}. Furthermore, The later four protocols inherit the differential privacy property of the function they compute.

Forth, we provide a complexity bound of multiparty computation of randomized function, which show the intrinsic complexity of the method the paper use to achieve obliviousness, i.e., bits $\mathrm{XOR}$. 

Finally, to show that the protocols in Section \ref{sec-random-generation-protocols} are powerful and fundamental, we constructed a differentially private empirical risk minimization (ERM) protocol in the distributed setting by using the protocols in Section \ref{sec-random-generation-protocols}.


\subsection{Outline}

The rest of the paper is organized as follows: Section \ref{sec-preliminary} briefly reviews the Shamir's secret sharing scheme, differential privacy definition and non-uniform random variate generation. Section \ref{sec-randomized-security} discusses the security of the protocol computing randomized function. Section \ref{sec-mdp} mainly discusses how can a protocol inherit the differential privacy property of a function it computes. The composition theorem of differentially private protocols is also given. Section \ref{sec-random-generation-protocols} constructs some fundamental protocols to generate random variates in the distributed setting. It also provides the Gaussian mechanism, the Laplace mechanism and the exponential mechanism in the distributed setting. Section \ref{sec-application} applies the protocols in Section \ref{sec-random-generation-protocols} to solve the empirical risk minimization problem. Section \ref{sec-related-work} presents related works. Finally, concluding remarks and a discussion of future work are presented in Section \ref{sec-conclusion}.

\section{Preliminary} \label{sec-preliminary}

\subsection{Secure Multiparty Computation Framework}  \label{sec-shamir}


MPC enables $n $ parties $P_1, \ldots, P_n$ jointly evaluate a prescribed function on private inputs in a privacy-preserving way. We assume that the $n$ parties are connected by perfectly secure channels in a synchronous network. We employ the $(t,n)$-Shamir's secret sharing scheme for representation of and secure computation on private values, by using which the computation of a function $f(\cdot)$ can be divided into three stages. Stage \uppercase\expandafter{\romannumeral1}: Each party enters his input $x_i$ to the computation using Shamir's secret sharing. Stage \uppercase\expandafter{\romannumeral2}: The parties simulate the circuit computing $f(x_1, \ldots, x_n)$, producing a new shared secret $T$ whose value is $f(x_1, \ldots, x_n)$. Stage \uppercase\expandafter{\romannumeral3}: At least $t+1$ shares of $f(x_1, \ldots, x_n)$ are sent to one party, who reconstructs it. All operations are assumed to be performed in a prime field $\mathbb F_p$. When treating fixed point and floating point number operations, we borrow the corresponding protocols in \cite{DBLP:conf/fc/CatrinaS10,DBLP:conf/ndss/AliasgariBZS13,DBLP:conf/acsac/EignerMPPK14}. By using these protocols we can treat the real number operations in a relatively precise way. Therefore, in the paper we assume there are some fundamental real number operations in MPC: addition, multiplication, division, comparison, exponentiation etc. For more formal and general presentation of this approach please see \cite{DBLP:conf/eurocrypt/CramerDM00,DBLP:conf/stoc/Ben-OrGW88}.

\subsection{Differential Privacy}  \label{sub-sec-dp}

Differential privacy of a function means that any change in a single individual input may only induce a small change in the distribution on its outcomes. A differentially private function is a kind of randomized function. The related definitions follow from the book \cite{DBLP:journals/fttcs/DworkR14}.

\begin{definition}[Randomized Function] \label{definition-randomized}
A randomized function $\mathcal M$ with domain $A$ and discrete range $B$ is associated with a mapping $M: A \rightarrow \Delta(B)$, where $\Delta(B)$ denotes the set of all the probability distribution on $B$. On input $x \in A$, the function $\mathcal M$ outputs $\mathcal M(x) = b$ with probability $ (M(x))_b$ for each $b \in B$. The probability space is over the coin flips of the function $\mathcal M$.
\end{definition}

\begin{definition}[Differential Privacy \cite{DBLP:conf/icalp/Dwork06,DBLP:journals/fttcs/DworkR14}] \label{definition-dp}
A randomized function $\mathcal{M}$ gives \emph{$(\epsilon,\delta)$-differential privacy} if for all datasets $x$ and $y$ differing on at most one element, and all $S \subset Range(\mathcal{M})$,
\[ \Pr[\mathcal{M}(x) \in S] \le \exp(\epsilon) \times \Pr[\mathcal{M}(y) \in S] + \delta, \]
where the probability space is over the coin flips of the function $\mathcal{M}$. If $\delta=0$, we say that $\mathcal{M}$ is \emph{$\epsilon$-differentially private}.
\end{definition}

There are mainly two ways to achieve differential privacy: \emph{noise mechanism} \cite{DBLP:conf/icalp/Dwork06} and \emph{exponential mechanism} \cite{DBLP:conf/focs/McSherryT07}. Noise mechanism computes the desired function on the data and then adds noise proportional to the maximum change than can be induced by changing a single element in the data set.


\begin{definition}[\cite{DBLP:journals/fttcs/DworkR14}] \label{definition-laplace-gaussian-exponential}
The exponential mechanism $\mathcal{M}(x, u, \mathcal R)$ outputs an element $r \in \mathcal R$ with probability proportional to $\exp(\frac{\epsilon u(x,r)}{2\Delta u})$. The Gaussian mechanism $\mathcal{M}(x, f)$ generates a random vector $r =(r_1, \ldots, r_n)$, where each $r_i \sim \mathcal N(f_i(x),\sigma^2)$, $\sigma >\sqrt{2 \ln 1.25/\delta}\Delta_2 f/\epsilon$. The Laplace mechanism $\mathcal{M}(x, f)$ generates a random vector $r =(r_1, \ldots, r_n)$, where each $r_i \sim \mathcal \mathrm{Lap}(f_i(x),\Delta f/\epsilon)$, $\mathrm{Lap}(f_i(x),\Delta f/\epsilon)$ denotes the Laplace distribution with variance $2(\Delta f/\epsilon)^2$ and mean $f_i(x)$.

Both the exponential mechanism and the Laplace mechanism satisfy $\epsilon$-differential privacy. The Gaussian mechanism satisfies $(\epsilon,\delta)$-differential privacy.
\end{definition}

Any sequence of computations that each provide differential privacy in isolation also provide differential privacy in sequence.

\begin{lemma}[Sequential composition of DP \cite{DBLP:conf/sigmod/McSherry09}] \label{lemma-composition-dp}
Let $\mathcal M_i$ is $(\epsilon_i, \delta_i)$-differentially private. Then their combination, defined to be $\mathcal M_{1 \cdots n}(x)= (\mathcal M_1, \ldots, \mathcal M_n)$, is $(\sum_i\epsilon_i, \sum_i\delta_i)$-differentially private.
\end{lemma}

Note that Lemma \ref{lemma-composition-dp} is true not only when $\mathcal M_1, \ldots, \mathcal M_n$ are run independently, but even when subsequent computations can incorporate the outcomes of the preceding computations \cite{DBLP:conf/sigmod/McSherry09}.



\subsection{Non-Uniform Random variate Generation}

Non-uniform random variate generation studies how to generate random variates drawn from a prescribed distribution function. In general, it assumes that there exists a random variate, called it a \emph{seed}, to generate \emph{randomness} for the random variates needed.



\subsubsection{The Inversion Method} \cite{Devroye86non-uniformrandom} is an important method to generate random variates, which is based upon the following property:

\begin{theorem} \label{theorem-inversion}
Let $F$ be a continuous distribution function on $\mathbb R$ with inverse $F^{-1}$ defined by
\[ F^{-1}(u) = \inf \{x:F(x)=u, 0<u <1 \}. \]
If $U$ is a uniform $[0,1]$ random variate, then $F^{-1}(U)$  has distribution function $F$. Also, if $X$ has distribution function $F $, then $F ( X )$ is uniformly distributed
on $[0,1]$.
\end{theorem}

Theorem \ref{theorem-inversion} \cite[Theorem 2.1]{Devroye86non-uniformrandom} can be used to generate random variates with an arbitrary univariate continuous distribution function $F$ provided that $F^{-1}$ is explicitly known. Formally, we have

\begin{algorithm} \label{algorithm-inversion}
\SetKwInOut{Input}{input}\SetKwInOut{Output}{output}
\Input{None}
\Output{A random variate drawn from $F$}
\SetAlgorithmName{Algorithm} \;

Generate a uniform $[0,1]$ random variate $U $\;
\Return $X \gets F^{-1}(U)$.
\caption{The inversion method}

\end{algorithm}

\subsubsection{The Gibbs Sampling} \cite{athreya2006measure} is one Markov Chain Monte Carlo (MCMC) algorithm, a kind of algorithms widely used in statistics, scientific modeling, and machine learning to estimate properties of complex distributions. For a distribution function $F$, the Gibbs sampling generates a Markov chain $\{Y_m\}_{m \ge 0}$ with $F$ as its stationary distribution.

Let $f(r_1, \ldots, r_k)$ be the density function of $F$ and let $(R_1, \ldots, R_k)$ be a random vector with distribution $F$. For $r = (r_1, r_2, \ldots, r_{k})$, let $r_{(i)} = (r_1, r_2, \ldots, r_{i-1}, r_{i+1}, \ldots, r_{k})$ and $p_i(\cdot | r_{(i)})$ be the conditional density of $R_i$ given $R_{(i)}=r_{(i)}$. Algorithm \ref{algorithm-Gibbs-sampling} generates a Markov chain $\{Y_m\}_{m \ge 0}$.

\begin{algorithm} \label{algorithm-Gibbs-sampling}
\SetKwInOut{Input}{input}\SetKwInOut{Output}{output}
\Input{Set the initial values $[R_{0j}] \gets [r_{0j}] , j= 1,2, \ldots, k-1$}
\Output{A random vector $[Y]$ with density $f(r)$}
\SetAlgorithmName{Algorithm} \;

Generate a random variate $[R_{0k}]$ from the conditional density $p_k(\cdot | R_{\ell} = r_{0\ell}, \ell= 1,2, \ldots, k-1)$\;
\For{$i :=1$ \textbf{to} $m$} {
    \For{$j :=1$ \textbf{to} $k$}{
    Generate a random variate $[R_{ij}]$ from the conditional density $ p_j(\cdot | R_{\ell} = s_{i \ell}, \ell \in \{ 1, \ldots, k \}\setminus \{j\})$, where $s_{i \ell}= r_{i \ell}$ for $1 \le \ell <j$ and $ s_{i \ell}= r_{(i-1) \ell}$ for $j < \ell \le k$;
    }
}
The parties output the random vector $[Y_m] = ([R_{m1}], \ldots, [R_{mk}])$.
\caption{The Gibbs sampling algorithm}
\end{algorithm}

\subsection{Notations}

Throughout the paper, let $[x]$ denote that the value $x$ is secretly shared among the parties by using Shamir's secret sharing. Let $s \sim F$ denote the random variate $s$ follows the distribution function $F$.

\section{The Security of Computing Randomized Function} \label{sec-randomized-security}

In the section, we study the security of computing randomized function in the distributed setting. We focus on in what condition can the computation of a randomized function be reduced to a deterministic one. The results of the section is vital to construct differentially private protocols.

We first give the definition of (statistically) indistinguishability.

\begin{definition}[Indistinguishability \cite{DBLP:books/cu/Goldreich2004}]
Two probability ensembles $X \stackrel{\mathrm{def} }{=}\{X_w\}_{w \in S}$ and $Y \stackrel{\mathrm{def} }{=}\{Y_w\}_{w \in S}$ are called (statistically) indistinguishable, denoted $X \equiv Y$, if for every positive polynomial $p(\cdot)$, every sufficiently large $k$, and every $w \in S \cap \{0, 1\}^k$, it holds that

\[  \sum_{\alpha \in \{0,1\}^*} |\Pr[X_w = \alpha] - \Pr[Y_w = \alpha]|< \frac{1}{p(k)}.\]
\end{definition}



The security definition of protocols computing randomized functions mainly follows from \cite[Definition 7.5.1]{DBLP:books/cu/Goldreich2004}.

\begin{definition} \label{definition-security-mpc}
Let $\mathcal M(x)$ be an $n$-ary \emph{randomized} function and let $\pi(x)$ be an $n$ party protocol to compute $\mathcal M(x)$, where $\mathcal M_i(x)$ denotes the $i$-th element of $\mathcal M(x)$. The \emph{view} of the party $P_i$ during an execution of
$\pi$ on $(x,s)$, denoted $\mathrm{VIEW}^{\pi}_i(x,s)$, is $(x_i, s_i, m_1,\ldots, m_t )$, where $s_i \sim F_i$ is a random variate $P_i$ inputs, and $m_j$ represents the $j$-th message it has received. The output of $P_i$ after an execution of $\pi$ on $(x,s)$, denoted $\mathrm{OUTPUT}^{\pi}_i(x,s)$, is implicit in the party's own view of the execution, and $\mathrm{OUTPUT}^{\pi}(x,s) = (\mathrm{OUTPUT}^{\pi}_1(x,s), \ldots, \mathrm{OUTPUT}^{\pi}_n(x,s))$. For $I = \{i_1, \ldots, i_k\} \subseteq \{1, \ldots, n\}$, let $\mathcal M_I(x)$ denote the subsequence $\mathcal M_{i_1}(x), \ldots, \mathcal M_{i_k}(x)$. Let $\mathrm{VIEW}^{\pi}_I(x,s) \stackrel{\mathrm{def}}{=} \{I, \mathrm{VIEW}^{\pi}_{i_1}(x,s), \ldots, \mathrm{VIEW}^{\pi}_{i_k}(x,s)\}$. We say that $\pi$ \emph{privately computes} $\mathcal M$ if there exists an algorithm $S$, such that for every $I \subseteq \{1, \ldots, n\}$, it holds that
    \begin{equation} \label{equation-mpc}
     \{ S( I, x_I, F_I, \mathcal{M}_I(x)), \mathcal{M}(x) \}_x \equiv \{\mathrm{VIEW}^{\pi}_I(x,s), \mathrm{OUTPUT}^{\pi}(x,s)\}_x,
    \end{equation}
where $x = (x_1, \ldots, x_n)$, $x_I=(x_{i_1}, \ldots, x_{i_k})$, $(x, s) = ((x_1,s_1), \ldots, (x_n,s_n))$ and $F_I=(F_{i_1}, \ldots, F_{i_k})$.

Throughout the paper, we assume that $\mathcal M_1(x)= \cdots = \mathcal M_n(x)$ and that $\mathrm{OUTPUT}_1^{\pi}(x,s) = \cdots = \mathrm{OUTPUT}_n^{\pi}(x,s)$. That is, each party obtains the same output.
\end{definition}


We remark that the above definition is slightly different from Definition 7.5.1 in \cite{DBLP:books/cu/Goldreich2004} that a private random variate $s= (s_1, \ldots, s_n)$ is input during the execution of $\pi$. The role of $s$ is to generate \emph{randomness} in order to compute $\mathcal M(x)$ (since $\mathcal M(x)$ is a randomized function). We call $s$ a \emph{seed} to compute the randomized function $\mathcal M(x)$. By providing the seed $s$, Definition \ref{definition-security-mpc} try to capture the vital characteristic of the process of computing randomized function in the distributed setting, such as Example \ref{example1} and Example \ref{example3}.

We define the notion of private reduction and cite a corresponding composition theorem. We refer the reader to \cite{DBLP:books/cu/Goldreich2004,DBLP:series/isc/HazayL10} for further details.

\begin{definition}[Privacy Reductions]
An oracle aided protocol using an oracle functionality $f$ privately computes $\mathcal M$ if there exists a simulator $S$ for each $I$ as in Definition \ref{definition-security-mpc}. The corresponding views are defined in the natural manner to include oracle answers. An oracle-aided protocol privately reduces $\mathcal M$ to $f$ if it privately computes $\mathcal M$ when using oracle functionality $f$.
\end{definition}

\begin{theorem}[Composition Theorem for the Semi-Honest Model\cite{DBLP:books/cu/Goldreich2004}] \label{theorem-mpc-composition}
Suppose $\mathcal M$ is privately reducible to $f$ and there exists a protocol for privately computing $f$. Then, the protocol defined by replacing each oracle-call to $f$ by a
protocol that privately computes $f$ is a protocol for privately computing $\mathcal M$.
\end{theorem}




\subsection{Reducing Computation of Randomized Function to Deterministic One}

Given a randomized function $\mathcal M$, let $\mathcal M( x, s')$ denote the value of $\mathcal M(x)$ when using a random seed $s'$ drawn from a distribution function $F$. That is, $\mathcal M(x)$ is the randomized process consisting of selecting $s' \sim F$, and deterministically computing $\mathcal M(x,s')$. Let $f$ be a deterministic function such that 
    \begin{equation} \label{equation-oracle-adid-protocol}
f ((x_1, s_1),\ldots, (x_n, s_n)) \stackrel{\mathrm{def}}{=} \mathcal M(x, g( s)),
\end{equation}
where $g$ is a deterministic function such that $s' = g(s)$, $s=(s_1, \ldots, s_n)$ and the random variate $s_i \sim F_i$. That is, in the distributed setting, we reduce computing the randomized function $\mathcal M$ to computing the deterministic function $f$. In the section, we consider the security problem induced by the reduction.

We now introduce the notion of \emph{obliviousness}, which is important to privately reduce the computation of randomized function to deterministic one.


\begin{definition}[Obliviousness] \label{definition-independence-auxiliary-input}
With the notation denoted as Definition \ref{definition-security-mpc}, the seed $s$ is said to be \emph{oblivious} to $\pi$ if for every $I = \{i_1, \ldots, i_k\} \subset \{1, \ldots, n\}$ and every $s_I'$, there is
    \[ \{ \mathrm{OUTPUT}^{\pi}(x,s)|s_I =s_I' \}_x \equiv \{ \mathrm{OUTPUT}^{\pi}(x,s) \}_x, \]
where $s_I =(s_{i_1}, \ldots, s_{i_k})$ and $s_I' $ is one admissible assignment to $s_I$.
\end{definition}

\begin{lemma} \label{lemma-obliviousness-contrary}
With the notation denoted as Definition \ref{definition-security-mpc}, if $s$ is not oblivious to $\pi$, then $\pi$ is not secure to compute $\mathcal M$.
\end{lemma}

\begin{proof}
Assume that $s$ is not oblivious to $\pi$. There then exist one $I$ and one $s_I'$ such that
    \[ \{ \mathrm{OUTPUT}^{\pi}(x,s)|s_I =s_I' \}_x \not\equiv \{ \mathrm{OUTPUT}^{\pi}(x,s) \}_x. \]
Now imaging the following execution of $\pi$: The parties $P_I$ input fixed value $s_I'$ for $s_I$. For any simulator with input $(x_I, M_I,F_I)$, who does not know $s_I =s_I'$, it is unable to get the distribution function $\{ \mathrm{OUTPUT}^{\pi}(x,s)|s_I =s_I' \}_x $. Therefore, there exist one $I$ and one $s_I'$ such that $\{ \mathrm{OUTPUT}^{\pi}(x,s)|s_I =s_I' \}_x $ is unable to be simulated by any simulator. However, the above distribution function is known to $P_I$ since they know the value $s_I =s_I'$, which implies that $ (\mathrm{OUTPUT}^{\pi}(x,s)|s_I =s_I')  \in \mathrm{VIEW}_I(x,s)$. Therefore, Equation (\ref{equation-mpc}) does not hold for $I$ and $s_I =s_I'$, for any simulator. Hence, $\pi$ is not secure to compute $\mathcal M$.

The claim is proved.
\end{proof}

\emph{Obliviousness}, which is a (no trivial) generalization of the notion ``obliviously'' in \cite{DBLP:conf/ccs/BunnO07}, says that the \emph{seed} (of each party or each proper subgroup of the parties) should be independent to the protocol's output. In other words, the execution of the protocol should be ``oblivious'' to the seed. One can verify that both the (not secure) two protocols in Example \ref{example1} to generate Gaussian noise and the (not secure) protocol in Example \ref{example3} to generate Laplace noise do not satisfy the property of \emph{obliviousness}.

Lemma \ref{lemma-obliviousness-contrary} gives a necessary condition to the security of protocol computing randomized function. Therefore, in order to reducing the computation of a randomized function to deterministic one, the seed should not only be \emph{secret} among the parties but also be \emph{oblivious} to the protocol's output. In the following, we give it a sufficient condition.

\begin{lemma} \label{lemma-reducing}
Let $\mathcal M$, $s$ and $f$ be defined as in Equation \ref{equation-oracle-adid-protocol}. Suppose that the following protocol, denoted $\pi$, is oblivious to $s$. Then it privately reduces $\mathcal M$ to $f $.
\begin{algorithm}
\SetKwInOut{Input}{input}\SetKwInOut{Output}{output}
\Input{$P_i$ gets $x_i$}
\Output{Each party outputs the oracle's response}
Step 1: $P_i$ selects $s_i \sim F_i$\;
Step 2: $P_i$ invokes the oracle of $f$ with query $(x_i , s_i )$, and records the oracle response.
\caption{privately reducing a randomized function to a deterministic one}
\end{algorithm}

\end{lemma}

\begin{proof}
Clearly, this protocol computes $\mathcal M$. To show that $\pi$ privately computes $\mathcal M$, we need to present a simulator $S_I$ for each group of parties $P_{i_1}, \ldots, P_{i_k}$'s view. For notational simplicity, we only prove that there exists a simulator $S_i$ for each party $P_i$. On input $(x_i , v_i )$, where $x_i$ is the local input to $P_i$ and $v_i$ is its local output, the simulator selects $s_i \sim F_i$, and outputs $(x_i , s_i , v_i )$. The main observation underlying the analysis of this simulator is that for every fixed $x = (x_1, \ldots, x_n)$ and $s'$, we have $v = \mathcal M(x,s')$ if and only if $v = f (x,s)$, for every $s$ satisfying $s' = g(s)$. Now, let $\xi_i$ be a random variable representing the random choice of $P_i$ in Step 1, and $\xi_i'$ denote the corresponding choice made by the simulator $S_i$. Then, referring to Equation \ref{equation-mpc}, we show that for every fixed $x$, $s_i$ and $v = (v_1, \ldots, v_n)$, it holds that
\begin{align*}
 & \Pr \left[ \mathrm{VIEW}^{\pi}_i(x,s)=(x_i,s_i, v_i) \wedge \mathrm{OUTPUT}^{\pi}(x,s) =v  \right]  \\
   &= \Pr \left[ (\xi_i = s_i) \wedge \mathrm{OUTPUT}^{\pi}(x,\xi)=f(x,\xi) =v  \right] \\
    &= \Pr \left[ (\xi_i = s_i)\right] \Pr \left[ \mathcal M(x) =v  \right]        \\
    &= \Pr \left[ (\xi_i' = s_i)\right] \Pr \left[ \mathcal M(x) =v  \right] \\
    &= \Pr \left[ (\xi_i' = s_i) \wedge \mathcal M(x) =v  \right]          \\
    &= \Pr \left[ S_i(x_i, F_i, \mathcal M_i(x)) =(x_i,s_i, v_i) \wedge \mathcal M(x) =v  \right]
\end{align*}
where the equalities are justified as follows: the 1st by the definition of $\pi$, the 2nd by the obliviousness of $\pi$ to $\xi$ and the definition of $f$, the 3rd by definition of $\xi_i$ and $\xi_i'$, the 4th by the independence of $\xi'$ and $\mathcal M$, and the 5th by definition of $S_i $. Thus, the simulated view (and output) is distributed identically to the view (and output) in a real execution.

Similarly, for each group of parties $P_{i_1}, \ldots, P_{i_k}$'s view, there exists a simulator $S$ such that Equation \ref{equation-mpc} holds.

The proof is complete.
\end{proof}

We remark that the proof technique in Lemma \ref{lemma-reducing} is borrowed from \cite[Proposition 7.3.4]{DBLP:books/cu/Goldreich2004}.

By combining Theorem \ref{theorem-mpc-composition}, Lemma \ref{lemma-obliviousness-contrary} and Lemma \ref{lemma-reducing}, we have the following theorem.

\begin{theorem}  \label{theorem-obliviousness-mpc-security}
Let Equation \ref{equation-oracle-adid-protocol} hold and let $\pi(x,s)$ be a secure protocol to compute $f(x,s)$. Then $\pi(x, \cdot)$ privately compute $\mathcal M(x)$ if and only if $\pi(x,s)$ is oblivious to $s$.
\end{theorem}

Theorem \ref{theorem-obliviousness-mpc-security} holds for differentially private functions since the later is a kind of randomized functions. Therefore, Theorem \ref{theorem-obliviousness-mpc-security} gives a necessary and sufficient condition about how to privately compute differentially private functions. Theorem \ref{theorem-obliviousness-mpc-security} can not only be used to determine whether a protocol computing differentially private function is secure, such as the protocols in Example \ref{example1} and Example \ref{example3}, but also be used to construct secure one, such as those protocols in Section \ref{sec-random-generation-protocols}.

\section{Multiparty Differential Privacy}  \label{sec-mdp}

For an $(\epsilon,\delta)$-differentially private function and a protocol computing it in the distributed setting, we are willing to see that the protocol has inherited the $(\epsilon,\delta)$-differential privacy property of the function it computes. It is intuitive that if the protocol \emph{privately} compute the function it will inherit the property naturally. In the section, we will prove that this is the fact.


We first introduce the notion of differential privacy in the distributed setting, which says that the view of each party (or each subgroup of the parties) is differentially private in respect to other parties' inputs.

\begin{definition}[Multiparty differential privacy \cite{DBLP:conf/crypto/GoyalMPS13}] \label{definition-mdp}
Let the notations be denoted as Definition \ref{definition-security-mpc}. We say that $x= (x_1, \ldots, x_n)$ and $y= (y_1, \ldots, y_n)$ differ on at most one element if there exists $i_0$ such that $x_i = y_i$ for all $i \in \{1,\ldots, n\}\setminus \{i_0\}$ and that $x_{i_0}, y_{i_0}$ differ on at most one element.
The protocol $\pi$ is said to be \emph{$(\epsilon, \delta)$-differentially private} if for all datasets $x,y$ differing on at most one element, for all $S$, and for all $I \subseteq \{1,2,\ldots, n\}\setminus \{i_0\}$,
    \begin{equation}  \label{equation-mdp}
       \Pr \left[ (\mathrm{VIEW}^{\pi}_I(x,s), \mathrm{OUTPUT}_I^{\pi}(x,s))  \in S \right]\le \exp(\epsilon) \times \Pr \left[ (\mathrm{VIEW}^{\pi}_I(y,s), \mathrm{OUTPUT}_I^{\pi}(y,s)) \in S \right] + \delta.
    \end{equation}
\end{definition}

\begin{theorem} \label{theorem-mdp}
Assume that $\mathcal M$ is an $(\epsilon, \delta)$-differentially private algorithm and that $\pi$ is a protocol to privately compute $\mathcal M$ in the distributed setting. Then $\pi$ is $(\epsilon, \delta)$-differentially private.
\end{theorem}

\begin{proof} 
For notational simplicity, we only prove the case of $I \in \{1,2,\ldots, n\}$. The general case can be treated similarly. We inherit the notations from Definition \ref{definition-mdp}.

Since $\mathcal M(x)$ is $(\epsilon, \delta)$-differentially private, we have
    \begin{equation} \label{equation-2}
     \Pr[\mathcal{M}(x) \in \bar S] \le \exp(\epsilon) \times \Pr[\mathcal{M}(y) \in \bar S] + \delta.
    \end{equation}
Then for all $S'$ and for all $i \in \{1,\ldots, n\}\setminus \{i_0\}$,
    \[ \Pr[( x_i, \mathcal{M}(x)) \in S'] \le \exp(\epsilon) \times \Pr[( y_i, \mathcal{M}(y)) \in S'] + \delta, \]
since $x_i = y_i$.

Therefore, for \emph{all} (post-processing \cite[page 229]{DBLP:journals/fttcs/DworkR14}) algorithm $S_i$ and all domain $S''$,
     \[ \Pr[S_i( x_i, \mathcal{M}(x)) \in S''] \le \exp(\epsilon) \times \Pr[S_i( y_i, \mathcal{M}(y)) \in S''] + \delta. \]
On the other hand, since $\pi(x)$ is a protocol to privately compute $\mathcal M(x)$, there \emph{exists} an algorithm $\bar S_i$ such that
    \[ \{ \bar S_i( x_i, \mathcal{M}(x)) \}_x \equiv \{\mathrm{VIEW}^{\pi}_i(x,s)\}_x. \]
Combining the last two formulas, we have
    \begin{equation} \label{equation-3}
      \Pr[\mathrm{VIEW}^{\pi}_i(x,s) \in S''] \le \exp(\epsilon) \times \Pr[\mathrm{VIEW}^{\pi}_i(y,s) \in S''] + \delta.
    \end{equation}
Moreover, since $\mathrm{OUTPUT}_i^{\pi}(x,s)$ is implicit in $\mathrm{VIEW}^{\pi}_i(x,s)$ (see Definition \ref{definition-security-mpc}), the later can be seen as a post-processing of the former. Therefore, \emph{for all} $x$,
    \begin{equation} \label{equation-5}
     \Pr \left[ (\mathrm{VIEW}^{\pi}_i(x,s), \mathrm{OUTPUT}_i^{\pi}(x,s))  \in S \right]= \Pr[\mathrm{VIEW}^{\pi}_i(x,s) \in S'']
    \end{equation}
Inputting Equation (\ref{equation-5}) into Equation (\ref{equation-3}), we have Equation (\ref{equation-mdp}).

The proof is complete.
\end{proof}

The following theorem provides the sequential composition property to differentially private protocols.

\begin{theorem}[Composition theorem] \label{theorem-composition-DP} Assume that the protocol $\pi_i$ privately computes  $(\epsilon_i, \delta_i)$-differentially private algorithm $\mathcal M_i$ for $1 \le i \le n$. Then their composition, defined to be $\mathcal \pi_{1 \cdots n}= (\pi_1, \ldots, \pi_n)$, is $(\sum_i\epsilon_i, \sum_i\delta_i)$-differentially private.





\end{theorem}

\begin{proof}
Since each $\pi_i$ is secure to compute $\mathcal M_i$, then their combination $\pi_{1 \cdots n}$ is secure to compute $\mathcal M_{1 \cdots n}$ by Theorem \ref{theorem-mpc-composition}. By Theorem \ref{theorem-mdp} and Lemma \ref{lemma-composition-dp} we have $\pi_{1 \cdots n}$ is $(\sum_i\epsilon_i, \sum_i\delta_i)$-differentially private.
\end{proof}

Note that, by Lemma \ref{lemma-composition-dp}, Theorem \ref{theorem-composition-DP} is true not only when $\pi_1, \ldots, \pi_n$ are run independently, but even when subsequent computations can incorporate the outcomes of the preceding computations.


\section{Protocol Construction} \label{sec-random-generation-protocols}

In this section, we use the results in Section \ref{sec-randomized-security} to construct secure protocols to compute randomized functions. We first design a protocol to generate the uniform random variate and a protocol to implement the inversion method in the distributed setting. Then we construct secure protocols to implement the Laplace mechanism and the exponential mechanism. Importantly, all of these protocols satisfy the property of \emph{obliviousness}.

Recall that, we let $[x]$ denote that the value $x$ is secretly shared among the parties by using Shamir's secret sharing.

\subsection{Multiparty Inversion Method}

We first provide Protocol \ref{protocol-Bernoulli-distributed} to generate random variate $X$ drawn from the Bernoulli $\mathrm{Bern}(1/2)$ distribution in the distributed setting, where $X$ takes on only two values: 0 and 1, both with probability $1/2$. Protocol \ref{protocol-Bernoulli-distributed} uses the fact that the $\mathrm{XOR}$ of two Bernoulli $\mathrm{Bern}(1/2)$ random variates is also a Bernoulli $\mathrm{Bern}(1/2)$ random variate.

\begin{algorithm}  \label{protocol-Bernoulli-distributed}
\SetKwInOut{Input}{input}\SetKwInOut{Output}{output}
\Input{None}
\Output{The parties obtain a random variate $[X]$ drawn from $\mathrm{Bern}(1/2)$}
The party $P_i$ generates a random bit $s_i$ drawn from the Bernoulli $\mathrm{Bern}(1/2)$ distribution by flipping an unbiased coin and shares it among the parties, for $ 1 \le i \le n$\;
The parties compute $[X] \gets \oplus_{i=1}^n [s_i]$ and output it, where $\oplus$ denote $\textrm{XOR}$ operation.
\caption{Multiparty generation of Bernoulli $\mathrm{Bern}(1/2)$ random variate}
\end{algorithm}

We give Protocol \ref{protocol-Gaussian-distributed} to generate random variate drawn from the standard Gaussian distribution $\mathcal{N}(0,1)$ in the distributed setting. The protocol approximates the Gaussian distribution $\mathcal{N}(0,1)$ by using the central limit theorem.

\begin{algorithm}  \label{protocol-Gaussian-distributed}
\SetKwInOut{Input}{input}\SetKwInOut{Output}{output}
\Input{None}
\Output{The parties obtain a random variate $[X]$ drawn from $\mathcal{N}(0,1)$}
The parties generate $k$ independent random variates $[s_1], \ldots, [s_k]$ drawn from the Bernoulli $\mathrm{Bern}(1/2)$ distribution by invoking Protocol \ref{protocol-Bernoulli-distributed}\;
The parties compute $[Y] \gets \sum_{i=1}^k [s_i]$\;
The parties compute $[X] \gets ([Y]-k/2)/(\sqrt k/2)$.
\caption{Multiparty generation of Gaussian $\mathcal{N}(0,1)$ random variate}
\end{algorithm}

We now use Protocol \ref{protocol-Gaussian-distributed} to design Protocol \ref{protocol-uniform-distributed} to generate random variate drawn from the uniform distribution $U(0,1)$ in the distributed setting. Protocol \ref{protocol-uniform-distributed} uses the result in Theorem \ref{theorem-inversion}.


\begin{algorithm}   \label{protocol-uniform-distributed}
\SetKwInOut{Input}{input}\SetKwInOut{Output}{output}
\Input{None}
\Output{The parties obtain a random variate $[X]$ drawn from $U(0,1)$}
The parties generate a random variate $[\xi]$ drawn from $\mathcal{N}(0,1)$ by using Protocol \ref{protocol-Gaussian-distributed} \;
The parties compute $[X] \leftarrow [G(\xi)]$, where $G(x)$ is the distribution function of $\mathcal{N}(0,1)$. \emph{Note that $[G(\xi)] = \frac{1}{2} + [\frac{1}{\sqrt{2\pi}} \int_{0}^{\xi} \exp \left( -\frac{t^2}{2} \right) dt]$ where the second summand can be evaluated as follows by using the composite trapezoidal method \cite{ascher2011first}. Set $f(t)=\frac{1}{\sqrt{2\pi}} \exp \left( -\frac{t^2}{2} \right) $ and $Y = \int_{0}^{\xi} f(t) dt$.}
    \begin{enumerate}
    \item The parties negotiate a step length $h$ and a positive integer $k$ such that $kh=1$\;
    \item Each party computes $t_i = hi$ for $i \in \{0,\ldots,k\}$ separately\;
    \item The parties compute $[t_i'] = t_i [\xi]$ for $i \in \{0,\ldots,k\}$\;
    \item The parties compute $f([t_i'])$ for $i \in \{0,\ldots,k\}$\;
    \item The parties compute $[Y] \gets h[\xi](f(0) +f([\xi])+2\sum_{i=1}^{k-1} f([t_i']) )/2.$
    \end{enumerate}
\caption{Multiparty generation of Uniform $U(0,1)$ random variate}
\end{algorithm}

The inversion method presented in Algorithm \ref{algorithm-inversion} is an important method to generate univariate random variable. We now give its new edition in the distributed setting as shown in Protocol  \ref{protocol-inversion-distributed}. Protocol \ref{protocol-inversion-distributed} is a powerful and fundamental protocol to construct other complex protocols, such as protocols of the Gaussian mechanism, the Laplace mechanism and the exponential mechanism as shown in the followings.

\begin{algorithm}  \label{protocol-inversion-distributed}
\SetKwInOut{Input}{input}\SetKwInOut{Output}{output}
\Input{The univariate continuous distribution function $F(t)$} 
\Output{The parties obtain a random variate $[X]$ drawn from $F(t)$}
The parties generate a random number $[\xi]$ drawn from $U(0,1)$ by using Protocol \ref{protocol-uniform-distributed} \;
The parties compute $[X] \gets F^{-1}([\xi])$. \emph{Note that $F^{-1}([\xi])$, if it has an explicit expression, can be computed by using the non-decreasing property of $F(t)$ and the comparison operator. When $F^{-1}([\xi])=t$ only has implicit integral expression, i.e., $[\xi] = \int_{-\infty}^t f(s)ds$, it can be computed as follows. }
    \begin{enumerate}
    \item The parties compute $[\xi'] \gets [\xi] -\int_{-\infty}^0 f(s)ds = \int_{0}^t f(s)ds$;
    \item The parties choose two values $[a],[b]$ such that $\int_{0}^a f(s)ds \le \xi' \le \int_{0}^b f(s)ds$ by using the non-decreasing property of $\int_{0}^t f(s)ds$ and the comparison operator;
    \item The parties evaluate $[t]$ in the equation $[\xi']= \int_{0}^t f(s)ds$ by using the bisection method \cite{ascher2011first} over the initial interval $[a,b]$;
    \item The parties set $[X] \gets [t]$;
    \end{enumerate}
\caption{Multiparty inversion method}
\end{algorithm}

\subsection{Multiparty Differentially Private Protocols}

We now use the protocols of generating random variates to construct protocols of the Laplace mechanism and the exponential mechanism.

\subsubsection{Multiparty Gaussian Mechanism}

We give Protocol \ref{protocol-Gaussian-distributed} to generate random variate drawn from the Gaussian distribution $\mathcal{N}(f(x),\sigma^2)$ in the distributed setting, which achieves Gaussian mechanism. The protocol approximates the Gaussian distribution $\mathcal{N}(0,\sigma^2)$ by using the central limit theorem \cite[Corollary 11.1.3]{athreya2006measure}.

\begin{algorithm} \label{protocol-Gaussian-mechanism-distributed}
\SetKwInOut{Input}{input}\SetKwInOut{Output}{output}
\Input{Each party $P_i$ shares his input $x_i$ among the parites}
\Output{The parties obtain a random variate random vector $[X] =([X_1], \ldots, [X_n])$ drawn from $\prod_{j=1}^n \mathcal N(f_j(x),\sigma^2)$}
\For{$j :=1$ \textbf{to} $n$} {
The parties generate $k$ independent random variates $[s_1], \ldots, [s_k]$ drawn from the Bernoulli $\mathrm{Bern}(1/2)$ distribution by invoking Protocol \ref{protocol-Bernoulli-distributed}\;
The parties compute $[Y_i] \gets \sigma [s_i]$ for $1 \le i \le k$\;
The parties compute $[\bar Y] \gets \sum_{i=1}^k [Y_i]/k$\;
The parties compute $[X'_j] \gets \sqrt k([\bar Y]-\sigma/2)$\;
The parties set $[X_j] \gets [f_j(x)] + [X'_j]$.
}
\caption{Multiparty Gaussian Mechanism}
\end{algorithm}

\subsubsection{Multiparty Laplace Mechanism}

The Laplace mechanism in the distributed setting is shown in Protocol \ref{protocol-laplace-mechanism}.

\begin{algorithm}  \label{protocol-laplace-mechanism}
\SetKwInOut{Input}{input}\SetKwInOut{Output}{output}
\Input{Each party $P_i$ secretly shares his input $x_i$ among the parties}
\Output{The parties obtain a random vector $[X] =([X_1], \ldots, [X_n])$ drawn from $\prod_{j=1}^n \mathrm{Lap}(f_j(x),\Delta f/\epsilon)$}
\For{$j :=1$ \textbf{to} $n$} {
The parties generate a random variate $[\xi_j]$ drawn from $\mathrm{Lap}(\Delta f/\epsilon)$ by using Protocol \ref{protocol-inversion-distributed}($F(t)$), where $F(t) = \frac{\epsilon}{2\Delta f} \int_{-\infty}^t \exp(-\frac{\epsilon|s|}{\Delta f} )ds$\;
The parties set $[X_j] \gets [f_j(x)] + [\xi_j]$.
}
\caption{Multiparty Laplace Mechanism}
\end{algorithm}

\subsubsection{Multiparty Exponential mechanism}   \label{sub-sec-exponential-mechanism}

When the range $\mathcal R$ is a finite set, we set $\mathcal R = \{ r_1, \ldots, r_{|\mathcal R|} \}$. The aim of the exponential mechanism is to draw a random element $r \in \mathcal R$ with probability $\exp(\frac{\epsilon u(x,r)}{2\Delta u})$. Protocol \ref{protocol-discrete-exponential-distributed} achieves the aim whose main idea is the sequential search algorithm in \cite[page 85]{Devroye86non-uniformrandom}. In Protocol \ref{protocol-discrete-exponential-distributed}, the comparison function $\mathrm{LT}([S], [\xi])=1$ if $S < \xi$ and $\mathrm{LT}([S], [\xi])=0$ if $S \ge \xi$.


\begin{algorithm} \label{protocol-discrete-exponential-distributed}
\SetKwInOut{Input}{input}\SetKwInOut{Output}{output}
\Input{Each party $P_i$ secretly shares his input $x_i$ among the parties}
\Output{The parties obtain a random variate $[X]$ on $ \mathcal R$ with probability mass function $\exp(\frac{\epsilon u(x,R)}{2\Delta u})$, where $x = (x_1, \ldots, x_n)$ and $\mathcal R = \{ 1,2,\ldots, |\mathcal R| \}$}
The parties compute $[p_{i}] \leftarrow [\exp(\frac{\epsilon u(x,i)}{2\Delta u})]$ for each $i \in \mathcal R$\;
The parties generate a random variate $[U]$ drawn from $U(0,1)$ by using Protocol \ref{protocol-uniform-distributed}\;
The parties compute $[\xi] \leftarrow [U]\times [\sum_{i \in \mathcal R} p_i]$\;
The parties set $[X] \gets 1$, $[S] \gets [p_1]$\;
\For{$k :=2$ \textbf{to} $|\mathcal R|$} {
    The parties compute $[X] \gets [X] + \mathrm{LT}([S], [\xi])$\;
    The parties compute $[S] \gets [S] + [p_i]$\;
}
The parties output $[X]$.
\caption{Multiparty discrete Exponential mechanism}
\end{algorithm}

When $\mathcal R $ is a set of high dimensional continuous random vectors, we can use the Gibbs sampling algorithm.


Let $\mathcal R = \{(r_1, \ldots, r_{k}) : r_i \in \mathbb R \mathrm{\; for \;} 1 \le i \le k\}$. Setting $\alpha = \int_{r \in \mathcal R} \exp(\frac{\epsilon u(x,r)}{2\Delta u})dr$, then $f(r) = \frac{1}{\alpha}\exp(\frac{\epsilon u(x,r)}{2\Delta u})$ is a density function on $\mathcal R$. The Gibbs sampling method generates a Markov chain $\{R_m\}_{m \ge 0}$ with $f(r)$ as its stationary density. Let

\[ p_i(\cdot | r_{(i)}) = \frac{f(r_1, r_2, \ldots, r_{i-1}, \cdot, r_{i+1}, \ldots, r_{k})}{ \int_{x \in \mathbb R} f(r_1, r_2, \ldots, r_{i-1}, x, r_{i+1}, \ldots, r_{k})dx } . \]

Note that $ p_i(\cdot | r_{(i)})$ is a univariate density function. Protocol \ref{protocol-gibbs-sampler} outputs a random vector $R$ drawn from the density $f(r)$, which uses the multiparty edition of Algorithm \ref{algorithm-Gibbs-sampling}.


\begin{algorithm} \label{protocol-gibbs-sampler}
\SetKwInOut{Input}{input}\SetKwInOut{Output}{output}
\Input{Each party $P_i$ secretly shares his input $x_i$ among the parties; The parties obtain the initial values $[R_{0j}] \gets [r_{0j}] , j= 1,2, \ldots, k-1$}
\Output{A random vector $R$ drawn from $f(r)$}
The parties generate a random variate $[R_{0k}]$ from the conditional density $p_k(\cdot | X_{\ell} = r_{0\ell}, \ell= 1,2, \ldots, k-1)$\;
\For{$i :=1$ \textbf{to} $m$} {
    \For{$j :=1$ \textbf{to} $k$}{
    The parties generate a random variate $[R_{ij}]$ from the conditional density $ p_j(\cdot | X_{\ell} = s_{i \ell}, \ell \in \{ 1, \ldots, k \}\setminus \{j\})$ by using Protocol \ref{protocol-inversion-distributed}, where $s_{i \ell}= r_{i \ell}$ for $1 \le \ell <j$ and $ s_{i \ell}= r_{(i-1) \ell}$ for $j < \ell \le k$;
    }
}
The parties output the random vector $[R_m] = ([R_{m1}], \ldots, [R_{mk}])$.
\caption{The multiparty high dimensional exponential mechanism}
\end{algorithm}

\subsection{Security and Privacy Analysis}

The security of the protocols in Section \ref{sec-random-generation-protocols} can be analyzed by Lemma \ref{lemma-reducing}. Given a randomized function $\mathcal M$, we first select $s' \sim F$, and deterministically compute $\mathcal M(x,s')$. Let $f$ be a deterministic function satisfying Equation \ref{equation-oracle-adid-protocol}. By Lemma \ref{lemma-reducing}, if there is a protocol $\pi$ privately computing $f$ and that $\pi(\cdot, s)$ is oblivious to $s \sim F$, then $\pi(x, \cdot)$ is secure to compute $\mathcal M(x)$ by inputting the seed $s \sim F$.

Since the paper focuses on computing randomized functions in the distributed setting and in order to keep the readability, the protocols in the section are presented in an algorithmic manner but not explicitly presented in the mathematical operations like \cite{DBLP:conf/acsac/EignerMPPK14}. The involved sub-protocols to compute some fundamental operations, e.g., addition, multiplication, $\mathrm{XOR}$, comparison and exponentiation etc., and the sub-protocols to compute some fundamental algorithms, e.g., the bisection method and the composite trapezoidal method, can be achieved by the works in \cite{DBLP:conf/ndss/AliasgariBZS13,DBLP:conf/fc/CatrinaS10,DBLP:conf/acsac/EignerMPPK14}, which would be as one future work. Therefore, in the paper, we assume that each deterministic function can be privately computed. Hence, to prove the security of the protocols in Section \ref{sec-random-generation-protocols}, we only need to prove the \emph{correctness} and \emph{obliviousness} of these protocols.

\subsubsection{Semi-Honest Model}

\emph{Obliviousness:} All the protocols in Section \ref{sec-random-generation-protocols} satisfy the property of obliviousness. This is because of seeds needed in these protocols are all input through invoking Protocol \ref{protocol-Bernoulli-distributed}. However, it can be easily verified that Protocol \ref{protocol-Bernoulli-distributed} satisfies the property of obliviousness. Therefore, other protocols inherit the obliviousness of Protocol \ref{protocol-Bernoulli-distributed}.

\emph{Correctness:} Protocol \ref{protocol-Bernoulli-distributed} is due to the fact that the $\mathrm{XOR}$ of two Bernoulli $\mathrm{Bern}(1/2)$ random variates is also a Bernoulli $\mathrm{Bern}(1/2)$ random variate. Therefore, $\oplus_{i=1}^n b_i$ is a $\mathrm{Bern}(1/2)$ random variate since each $b_i$ is a $\mathrm{Bern}(1/2)$ random variate. The correctness of Protocol \ref{protocol-Gaussian-distributed} is due to the central limit theorem \cite[Corollary 11.1.3]{athreya2006measure}. The correctness of Protocol \ref{protocol-uniform-distributed} is due to Theorem \ref{theorem-inversion}: If the random variate $\xi$ is drawn from  $\mathcal N(0,1)$, then $G(\xi) $ is drawn from $U(0,1)$, where $G(x) = \frac{1}{\sqrt{2\pi}} \int_{-\infty}^{x} \exp \left( -\frac{t^2}{2} \right) dt$ is the distribution function of $\mathcal N(0,1)$. The correctness of Protocol \ref{protocol-inversion-distributed} is due to the classical inversion method, i.e., Algorithm \ref{algorithm-inversion}. In Protocol \ref{protocol-Gaussian-mechanism-distributed}, Step 2 to Step 5 generate a random variate $X' \sim \mathcal{N}(0, \sigma^2)$ by using the central limit theorem \cite[Corollary 11.1.3]{athreya2006measure}. Then $f_j(x) + X'_j \sim \mathcal{N}(f_j(x),\sigma^2)$. Protocol \ref{protocol-laplace-mechanism} is due to the Laplace mechanism in Definition \ref{definition-laplace-gaussian-exponential}. Protocol \ref{protocol-discrete-exponential-distributed} is due to the sequential search algorithm in \cite[page 85]{Devroye86non-uniformrandom}. The correctness of Protocol \ref{protocol-gibbs-sampler} uses the correctness of Algorithm \ref{algorithm-Gibbs-sampling}.

\begin{corollary}
Protocol \ref{protocol-Gaussian-mechanism-distributed} is $(\epsilon,\delta)$-differentially private. Protocol \ref{protocol-laplace-mechanism}, Protocol \ref{protocol-discrete-exponential-distributed} and Protocol \ref{protocol-gibbs-sampler} are all $\epsilon$-differentially private.
\end{corollary}

\begin{proof}
This is a direct corollary of Theorem \ref{theorem-mdp}.
\end{proof}

\subsubsection{Malicious Model}

By forcing parties to behave in an effectively semi-honest manner, we can transform the above protocols in the semi-honest model into protocols secure in the malicious-behavior model. The above process needs some preliminaries: the commitment schemes, zero-knowledge proof techniques and the Verifiable Secret Sharing (VSS) scheme. In the paper we do not intend to give it a detailed construction but as a future work. Besides of these, we consider the malicious behavior in computing seeds. Seeds are generated bit by bit by invoking Protocol \ref{protocol-Bernoulli-distributed}, in which a malicious party may input either a non-bit random element or a non-uniform random bit. The first malicious behavior can be avoided by verifying the input $x$ satisfies $x^2=x$. The second malicious behavior can be solved by first generating a public random variate drawn from $\mathrm{Bern}(1/2)$ and then $\mathrm{XOR}$ it with the output of Protocol \ref{protocol-Bernoulli-distributed} by the fact that the $\mathrm{XOR}$ of two random bits is uniform so long as one of which is uniform.

\subsection{Optimal Complexity}

By Section \ref{sec-randomized-security}, each party $P_1$ should input a seed $s_i$, a random variate, to the protocol $\pi$ for computing a randomized function $\mathcal M$ to generate the randomness of the final output. We call $s=(s_1, \ldots, s_n)$ a seed of $\pi$ for computing $\mathcal M$ and call the number of bit in $s$, denoted $|s|$, the length of $s$. Each protocol in Section \ref{sec-random-generation-protocols} takes independent random bit sequence as its seed. Note that the length of the seed is an important indicator of the complexity of the protocol, the minimum length of the seed is of special interest.

We now discuss the minimum length of the seeds of all the protocols for generating independent random bits.

\begin{theorem} \label{theorem-independent-bits-optimality}
Let $\pi$ be a protocol to privately compute the randomized function $\mathcal M$ of generating random vector $v=(v_1, \ldots, v_k)$, where $v_1, \ldots, v_k \sim_{i.i.d} \mathrm{Bern}(1/2)$. Let $s_i= (s_{i1}, \ldots, s_{i\ell_i})$ be the seed of $P_i$, where each $s_{ij}$ denotes a bit. Then $\ell_i \ge k$ for $1 \le i \le n$.

Therefore, the protocol $\pi'$ of independent $k$ times execution of Protocol \ref{protocol-Bernoulli-distributed} has the shortest seed among all the protocols for privately computing $\mathcal M$.
\end{theorem}

\begin{proof}
Let notations be denoted as Equation \ref{equation-oracle-adid-protocol} and Definition \ref{definition-independence-auxiliary-input}. Since $\pi$ should satisfy obliviousness, then for each $i \in \{1,\ldots, n\}$ and each admissible value $s_{\bar \imath}'$ of $s_{\bar \imath}$, we have
    \[ \{ \mathrm{OUTPUT}^{\pi}(x,s)|s_{\bar \imath} =s_{\bar \imath}' \}_x \equiv \{ \mathrm{OUTPUT}^{\pi}(x,s) \}_x, \]
where $s_{\bar \imath} = (s_1, \ldots, s_{i-1}, s_{i+1}, \ldots, s_n)$.

For the $(n+1)$-ary deterministic function $\mathcal M(x, s)$, let $g_i(s_i): = \mathcal M(x, s|x=x',s_{\bar \imath} =s_{\bar \imath}')$ denote the univariate deterministic function about $s_i$ when $x=x', s_{\bar \imath} =s_{\bar \imath}'$. Then $g_i: \{ 0,1\}^{\ell_i} \rightarrow \{ 0,1\}^k$. Assume that $\ell_i < k$ and, without loss of generality, set $\ell_i =k-1$. Set $S= \{g_i(y): y \in  \{ 0,1\}^{k-1} \}$. Since $|S| \le 2^{k-1} < 2^{k}$, there would have at lease $2^{k-1}$ elements of $\{ 0,1\}^k$ not contained in $S$. Letting $g_i(s_i) = (Y_1\cdots Y_{k})$, where $Y_1, \ldots, Y_k \sim_{i.i.d} \mathrm{Bern}(1/2)$, for any $k$-bit sequence $y_1\cdots y_{k} \notin S$, we have 
    \[ \prod_{j=1}^{k} \Pr[Y_j=y_j] = \Pr[Y_1\cdots Y_{k} =  y_1\cdots y_{k}] =0.\]
Therefore, there exists one $j \in \{1, \ldots, k\}$, such that $\Pr[Y_j=y_j]=0$, which is contrary to the assumption that $(Y_1\cdots Y_{k})$ are random bits. Therefore, $\ell_i \ge k$ for all $1 \le i \le n$.

On the other hand, the length of the seed of $\pi'$ is $nk$. Therefore, it has the shortest seed among all the protocols privately compute $\mathcal M$.

The claim is proved. 
\end{proof}

Theorem \ref{theorem-independent-bits-optimality} shows one intrinsic bound on optimizing the complexity of those protocols for computing randomized functions by invoking Protocol \ref{protocol-Bernoulli-distributed} and shows that our protocols in the section reach the bound.

\subsection{Application to Empirical Risk Minimization} \label{sec-application}

Our protocols are fundamental and powerful to construct other complex differentially private protocols. We now use our protocols to construct a differentially private empirical risk minimization (ERM) protocol in the distributed setting.

We consider a differentially private (ERM) algorithm \cite[Algorithm 1]{DBLP:journals/jmlr/ChaudhuriMS11}. For Algorithm 1 in \cite{DBLP:journals/jmlr/ChaudhuriMS11}, we can add a noise vector to the output of $\arg\min_f J(f,\mathcal{D})$  in order to achieve differential privacy, where
\[ J(f, \mathcal{D}) = \frac{1}{k} \sum_{i=1}^k \ell(f(x_i),y_i) + \Lambda N(f). \]
If the added noise vector $\textbf b$ is drawn from $\frac{1}{\alpha}\exp(-\lambda||\textbf b||)$ the output satisfies differential privacy, where $\lambda =\frac{2}{k\Lambda\epsilon}$.

In the distributed setting, let the dataset $\mathcal D=\{ (x_j,y_j)\} $ is partitioned into $n$ parts $\mathcal D_1, \ldots, \mathcal D_n$, where the party $P_i$ owns $\mathcal D_i$. Each party $P_i$ first shares its dataset $\mathcal D_i$ among the parties. Then the parties approximately compute a share $[f]$ of the minimizer of $J(f, ([\mathcal D_1], \ldots, [\mathcal D_n]))$ by using a deterministic function evaluation protocol. (Since the paper focuses on randomized function evaluation protocols, we omit to construct the protocol of computing $[f]$.) The parties now use Protocol \ref{protocol-n-gamma-distributed} to generate a random vector $[X]$ drawn from $\frac{1}{\alpha}\exp(-\lambda||\textbf b||)$, where Protocol \ref{protocol-n-gamma-distributed} is a multiparty edition of the polar method in \cite[page 225]{Devroye86non-uniformrandom}. The parties then compute $[X+f]$. Finally, the parties recover and output $X+f$, which would be a differentially private ERM in the distributed setting.

\begin{algorithm} \label{protocol-n-gamma-distributed}
\SetKwInOut{Input}{input}\SetKwInOut{Output}{output}
\Input{None}
\Output{A random variate $[X]$ drawn from $ f (x_1,\ldots,x_d)=\frac{1}{\alpha}e^{-\lambda \sqrt{x_1^2+ \ldots + x_d^2}}$ }
    The parties generate i.i.d normal randoms $[N_1], \ldots, [N_d]$ by Protocol \ref{protocol-Gaussian-distributed}\;
    The parties compute a share of random vector $[X'] \leftarrow ([\frac{N_1}{S}], \ldots, [\frac{N_d}{S}])$, where $S \leftarrow \sqrt{N_1^2 + \cdots + N_d^2}$ \;
The parties generate a random variate $[R]$ drawn from the density $dV_dr^{d-1}g(r)$ $(r \ge 0)$ by using Protocol \ref{protocol-inversion-distributed}, where $V_d = \frac{\pi^{d/2}}{\Gamma(d/2+1)}$ and $g(x) = \frac{1}{\alpha}e^{-\lambda x}$  \;
The parties compute $ [X] \gets [RX']$.
\caption{Multiparty generation of random variate drawn from $ f (x_1,\ldots,x_d) = \frac{1}{\alpha}e^{-\lambda \sqrt{x_1^2+ \ldots + x_d^2}}$}
\end{algorithm}

\section{Related Work}  \label{sec-related-work}

\emph{Secure multiparty computation} \cite{DBLP:conf/ndss/AliasgariBZS13,DBLP:conf/fc/CatrinaS10,DBLP:series/isc/HazayL10,DBLP:books/cu/Goldreich2004} studies how to privately compute functions in the distributed setting. The computation of randomized function, such as random variate generation, is seldom studied in MPC. Until recently, the development of DP in the distributed setting makes the study of the computation of randomized functions necessary in MPC. Except the works mentioned in Section \ref{sec-introduction}, other former works are presented as follows.  

Proposition 7.3.4 in \cite{DBLP:books/cu/Goldreich2004} privately reduces computing randomized function to a deterministic one. However, it does not give criterion about what kind of seed, which is used to generates the randomness, is secure. That is, the criterion for how to determine a protocol computing a randomized function is secure is not given. Our conclusion gives a sufficient and necessary condition (Theorem \ref{theorem-obliviousness-mpc-security}) about it and therefore gives the criterion, i.e., obliviousness. Furthermore, obliviousness gives some clue on finding more (efficient) reduction protocols except the one in \cite[Proposition 7.3.4]{DBLP:books/cu/Goldreich2004}. Note also that the randomized functions the paper considers are confided to be $n$-ary functions having the same value for all components.

The notion of \emph{obliviousness} can be seen as a (no trivial) generation of the notion \emph{Obliviously} in \cite{DBLP:conf/ccs/BunnO07}. However, they have one major difference: \emph{Obliviously} emphasises on the independence of the seed to the execution of the protocol computing the randomized function, where as \emph{obliviousness} focus on the independence of the seed to the output of the protocol computing the deterministic function, to which the randomized function is privately reduced. The advantage of the later is that it separates the choosing of the seed from the execution of the protocol computing the deterministic function, which makes the design and the analysis of the protocol computing randomized function easy to do.




\cite{DBLP:conf/eurocrypt/DworkKMMN06} gives two protocols to generate Gaussian random variate and Laplace random variate in the distributed setting, which are used to compute differentially private summation functions. Although Protocol \ref{protocol-Gaussian-distributed} in our paper is similar with the one in \cite{DBLP:conf/eurocrypt/DworkKMMN06} to generate Gaussian random variate, our work focus mainly on the fundamental theory and fundamental tools to compute randomized functions in the distributed setting and is therefore different from theirs.

Random Value Protocol \cite{DBLP:conf/ccs/BunnO07} is a two-party protocol to generate uniform random integers from $\mathbb Z_N$ while keeping $N$ secret, which is used to approximately generate uniform random variate \cite{DBLP:journals/tdsc/MohammedAFD14,DBLP:conf/pet/AlhadidiMFD12} following $U(0,1)$ in the two-party setting. It satisfies obliviousness but is too complicated that we can not see a way to extend it to a multiparty one. Furthermore, the distributed exponential mechanism protocols in \cite{DBLP:journals/tdsc/MohammedAFD14,DBLP:conf/pet/AlhadidiMFD12} are two special instantiations of Protocol \ref{protocol-discrete-exponential-distributed}.

\cite{DBLP:conf/acsac/EignerMPPK14} presents a protocol to implement exponential mechanism, in which a sub-protocol is needed to generate uniform random variate drawn from the uniform distribution $U(0,1)$. In order to generate such uniform random variate, the parties first secretly generate a uniform $(\gamma+1)$-bit integer using the protocol $\mathrm{RandInt}(\gamma+1)$. Then this integer is considered to be fractional part of fixed point number, whose integer part is 0. Afterwards, the fixed point number is converted to floating point by a secure protocol, which is output as the final result. The above protocol to generate uniform random number has two drawbacks. First, the invoked protocol $\mathrm{RandInt}(\gamma+1)$, borrowed from \cite{damgaard2006unconditionally,DBLP:conf/tcc/CramerDI05}, generates a uniform random element in $\mathbb Z_p$ by the modular sum of the uniform random elements in $\mathbb Z_p$ generated by each of the parties. Note that the modular sum of two uniform random elements in $\mathbb Z_p$ is, in general, not a uniform random elements in $\mathbb Z_p$ \cite{DBLP:conf/ccs/BunnO07}. Therefore, $\mathrm{RandInt}(\gamma+1)$ (most probably) generates a non-uniform random $(\gamma+1)$-bit integer, which in turn leads to the non-uniformity of the one in \cite{DBLP:conf/acsac/EignerMPPK14}. Second, since $\gamma$ is predetermined, the random number generated may not get value from many sub-intervals of $[0,1]$, such as the sub-interval $(0, 2^{-\gamma-1})$. Therefore, strictly speaking, the above method may not generate a random number with range $[0,1]$. Of course, the uniform property in the range $[0,1]$ of the generated random number will be not satisfied.

\cite{DBLP:conf/crypto/GoyalMPS13,DBLP:conf/focs/McGregorMPRTV10} studies the accuracy difference in computing Boolean functions between the client-server setting and the distributed setting. \cite{DBLP:conf/crypto/MironovPRV09} introduces the notion of computational differential privacy in the two-party setting. \cite{DBLP:conf/crypto/BeimelNO08} studies the influence to the accuracy of computing binary sum, gap threshold etc., when both of differentially private analyses and the construction of protocol are considered simultaneously, which is contrary to the paradigm we use in which we first analyze a problem using differentially private algorithm and then construct corresponding protocol to compute it.



\emph{Differential privacy} is a rigorous and promising privacy model. Much works have been done in differentially private data analysis \cite{DBLP:journals/jmlr/ChaudhuriSS13,DBLP:conf/kdd/McSherryM09,DBLP:conf/stoc/DworkTT014,DBLP:journals/jmlr/ChaudhuriMS11,DBLP:conf/nips/ChaudhuriV13,DBLP:conf/nips/ChaudhuriM08,DBLP:conf/icml/0002T13,DBLP:journals/pvldb/ZhangZXYW12,DBLP:conf/sigmod/ZhangCPSX14}. Our work tries to extend these algorithms to the distributed setting. It constructs fundamental theory, such as Theorem \ref{theorem-mdp} and Theorem \ref{theorem-composition-DP}, and fundamental tools, such as the protocols in Section \ref{sec-random-generation-protocols}, about it.

\emph{Non-uniform Random variate generation} \cite{Devroye86non-uniformrandom} is a well developed field in computer science and statistics. It studies how to generate non-uniform random variate drawn from the prescribed distribution function. Some work of the paper studies secure random variate generation in the distributed setting. It redesigns the traditional random variate generation protocols to adapt to the distributed setting. Note that most powerful algorithms, such as the rejection method, are not fit for the distributed setting.

\section{Conclusion and Future Work} \label{sec-conclusion}

The paper tried to answer in what condition can a protocol inherit the differential privacy property of a function it computes and how to construct such protocol. We proved that the differential privacy property of a function can be inherited by the protocol computing it if the protocol privately computes it. Then a theorem provided the sufficient and necessary condition of privately computing a randomized function (and so differentially private function) from a deterministic one.  The above result can not only be used to determine whether a protocol computing differentially private function is secure, but also be used to construct secure one. In obtaining these results, the notion \emph{obliviousness} plays a vital role, which captures the key security problems when computing a randomized function from a deterministic one in the distributed setting. However, we can not prove the assertion that a protocol can not inherit the differential privacy property of the function it computes if the protocol does not satisfy obliviousness. We tend to a negative answer to the assertion.

The theoretical results in Section \ref{sec-randomized-security} and Section \ref{sec-mdp} is fundamental and powerful to multiparty differential privacy. By using these results, some fundamental differentially private protocols, such as protocols for Gaussian mechanism, Laplace mechanism and Exponential mechanism, are constructed in Section \ref{sec-random-generation-protocols}. By using these fundamental protocols, differentially private protocols for many complex problems, such as the empirical risk minimization problem, can be constructed with little effort. Therefore, our results can be seen as a foundation and a pool of necessary tools for multiparty differential privacy.

Furthermore, obliviousness is of independent interest to MPC. The deep meaning of it in the security of computing randomized function is still needed to be explored. Theorem \ref{theorem-independent-bits-optimality} shows the intrinsic complexity of the method the paper use to achieve obliviousness, i.e., bits $\mathrm{XOR}$. Finding other efficient method to achieve obliviousness is therefore an important topic to reduce protocols' complexity.

\section*{Acknowledgment}

The research is supported by the following fund: National Science and Technology Major Project under Grant No.2012ZX01039-004.

\bibliography{Numerical-methods,random-number-generation,OS-privacy,algebra,zero-knowledge-proof,secure-multi-com,proximity-test,location-social-network,location-proof,location-access-control-privacy,location-privacy,anonymity-metrics,information-theory-anonymity,vehicle-ad-hoc-network,rfid-sec-privacy,sensor-adhoc-network,anonymous-attestation-blacklist,differential-privacy}

\begin{thebibliography}{10}

\bibitem{DBLP:conf/nips/PathakRR10}
Manas~A. Pathak, Shantanu Rane, and Bhiksha Raj.
\newblock Multiparty differential privacy via aggregation of locally trained
  classifiers.
\newblock In {\em Advances in Neural Information Processing Systems 23: 24th
  Annual Conference on Neural Information Processing Systems 2010. Proceedings
  of a meeting held 6-9 December 2010, Vancouver, British Columbia, Canada.},
  pages 1876--1884, 2010.

\bibitem{DBLP:conf/kdd/GantaKS08}
Srivatsava~Ranjit Ganta, Shiva~Prasad Kasiviswanathan, and Adam Smith.
\newblock Composition attacks and auxiliary information in data privacy.
\newblock In {\em Proceedings of the 14th {ACM} {SIGKDD} International
  Conference on Knowledge Discovery and Data Mining, Las Vegas, Nevada, USA,
  August 24-27, 2008}, pages 265--273, 2008.

\bibitem{DBLP:conf/sp/NarayananS08}
Arvind Narayanan and Vitaly Shmatikov.
\newblock Robust de-anonymization of large sparse datasets.
\newblock In {\em 2008 {IEEE} Symposium on Security and Privacy (S{\&}P 2008),
  18-21 May 2008, Oakland, California, {USA}}, pages 111--125, 2008.

\bibitem{DBLP:books/cu/Goldreich2004}
Oded Goldreich.
\newblock {\em The Foundations of Cryptography - Volume 2, Basic Applications}.
\newblock Cambridge University Press, 2004.

\bibitem{DBLP:conf/crypto/BeimelNO08}
Amos Beimel, Kobbi Nissim, and Eran Omri.
\newblock Distributed private data analysis: Simultaneously solving how and
  what.
\newblock In {\em Advances in Cryptology - {CRYPTO} 2008, 28th Annual
  International Cryptology Conference, Santa Barbara, CA, USA, August 17-21,
  2008. Proceedings}, pages 451--468, 2008.

\bibitem{DBLP:conf/pldi/MardzielHKS12}
Piotr Mardziel, Michael Hicks, Jonathan Katz, and Mudhakar Srivatsa.
\newblock Knowledge-oriented secure multiparty computation.
\newblock In {\em Proceedings of the 2012 Workshop on Programming Languages and
  Analysis for Security, {PLAS} 2012, Beijing, China, 15 June, 2012}, page~2,
  2012.

\bibitem{DBLP:journals/fttcs/DworkR14}
Cynthia Dwork and Aaron Roth.
\newblock The algorithmic foundations of differential privacy.
\newblock {\em Foundations and Trends in Theoretical Computer Science},
  9(3-4):211--407, 2014.

\bibitem{DBLP:conf/icalp/Dwork06}
Cynthia Dwork.
\newblock Differential privacy.
\newblock In {\em ICALP (2)}, pages 1--12, 2006.

\bibitem{DBLP:conf/focs/McSherryT07}
Frank McSherry and Kunal Talwar.
\newblock Mechanism design via differential privacy.
\newblock In {\em FOCS}, pages 94--103, 2007.

\bibitem{DBLP:conf/sigmod/RastogiN10}
Vibhor Rastogi and Suman Nath.
\newblock Differentially private aggregation of distributed time-series with
  transformation and encryption.
\newblock In {\em SIGMOD Conference}, pages 735--746, 2010.

\bibitem{DBLP:journals/tdsc/MohammedAFD14}
Noman Mohammed, Dima Alhadidi, Benjamin C.~M. Fung, and Mourad Debbabi.
\newblock Secure two-party differentially private data release for vertically
  partitioned data.
\newblock {\em IEEE Trans. Dependable Sec. Comput.}, 11(1):59--71, 2014.

\bibitem{7286780}
S.~Goryczka and L.~Xiong.
\newblock A comprehensive comparison of multiparty secure additions with
  differential privacy.
\newblock {\em IEEE Transactions on Dependable and Secure Computing},
  PP(99):1--1, 2015.

\bibitem{DBLP:conf/fc/CatrinaS10}
Octavian Catrina and Amitabh Saxena.
\newblock Secure computation with fixed-point numbers.
\newblock In {\em Financial Cryptography and Data Security, 14th International
  Conference, {FC} 2010, Tenerife, Canary Islands, January 25-28, 2010, Revised
  Selected Papers}, pages 35--50, 2010.

\bibitem{DBLP:conf/ndss/AliasgariBZS13}
Mehrdad Aliasgari, Marina Blanton, Yihua Zhang, and Aaron Steele.
\newblock Secure computation on floating point numbers.
\newblock In {\em 20th Annual Network and Distributed System Security
  Symposium, {NDSS} 2013, San Diego, California, USA, February 24-27, 2013},
  2013.

\bibitem{DBLP:conf/acsac/EignerMPPK14}
Fabienne Eigner, Matteo Maffei, Ivan Pryvalov, Francesca Pampaloni, and Aniket
  Kate.
\newblock Differentially private data aggregation with optimal utility.
\newblock In {\em Proceedings of the 30th Annual Computer Security Applications
  Conference, {ACSAC} 2014, New Orleans, LA, USA, December 8-12, 2014}, pages
  316--325, 2014.

\bibitem{DBLP:conf/eurocrypt/CramerDM00}
Ronald Cramer, Ivan Damg{\aa}rd, and Ueli~M. Maurer.
\newblock General secure multi-party computation from any linear secret-sharing
  scheme.
\newblock In {\em Advances in Cryptology - {EUROCRYPT} 2000, International
  Conference on the Theory and Application of Cryptographic Techniques, Bruges,
  Belgium, May 14-18, 2000, Proceeding}, pages 316--334, 2000.

\bibitem{DBLP:conf/stoc/Ben-OrGW88}
Michael Ben{-}Or, Shafi Goldwasser, and Avi Wigderson.
\newblock Completeness theorems for non-cryptographic fault-tolerant
  distributed computation (extended abstract).
\newblock In {\em Proceedings of the 20th Annual {ACM} Symposium on Theory of
  Computing, May 2-4, 1988, Chicago, Illinois, {USA}}, pages 1--10, 1988.

\bibitem{DBLP:conf/sigmod/McSherry09}
Frank McSherry.
\newblock Privacy integrated queries: an extensible platform for
  privacy-preserving data analysis.
\newblock In {\em SIGMOD Conference}, pages 19--30, 2009.

\bibitem{Devroye86non-uniformrandom}
Luc Devroye.
\newblock {\em Non-Uniform Random Variate Generation}.
\newblock Springer-Verlag, Berlin, Heidelberg, New York, 1986.

\bibitem{athreya2006measure}
Krishna~B Athreya and Soumendra~N Lahiri.
\newblock {\em Measure theory and probability theory}.
\newblock Springer Science \& Business Media, 2006.

\bibitem{DBLP:series/isc/HazayL10}
Carmit Hazay and Yehuda Lindell.
\newblock {\em Efficient Secure Two-Party Protocols - Techniques and
  Constructions}.
\newblock Information Security and Cryptography. Springer, 2010.

\bibitem{DBLP:conf/ccs/BunnO07}
Paul Bunn and Rafail Ostrovsky.
\newblock Secure two-party k-means clustering.
\newblock In {\em Proceedings of the 2007 {ACM} Conference on Computer and
  Communications Security, {CCS} 2007, Alexandria, Virginia, USA, October
  28-31, 2007}, pages 486--497, 2007.

\bibitem{DBLP:conf/crypto/GoyalMPS13}
Vipul Goyal, Ilya Mironov, Omkant Pandey, and Amit Sahai.
\newblock Accuracy-privacy tradeoffs for two-party differentially private
  protocols.
\newblock In {\em Advances in Cryptology - {CRYPTO} 2013 - 33rd Annual
  Cryptology Conference, Santa Barbara, CA, USA, August 18-22, 2013.
  Proceedings, Part {I}}, pages 298--315, 2013.

\bibitem{ascher2011first}
Uri~M Ascher and Chen Greif.
\newblock {\em A First Course on Numerical Methods}, volume~7.
\newblock Siam, 2011.

\bibitem{DBLP:journals/jmlr/ChaudhuriMS11}
Kamalika Chaudhuri, Claire Monteleoni, and Anand~D. Sarwate.
\newblock Differentially private empirical risk minimization.
\newblock {\em Journal of Machine Learning Research}, 12:1069--1109, 2011.

\bibitem{DBLP:conf/eurocrypt/DworkKMMN06}
Cynthia Dwork, Krishnaram Kenthapadi, Frank McSherry, Ilya Mironov, and Moni
  Naor.
\newblock Our data, ourselves: Privacy via distributed noise generation.
\newblock In {\em Advances in Cryptology - {EUROCRYPT} 2006, 25th Annual
  International Conference on the Theory and Applications of Cryptographic
  Techniques, St. Petersburg, Russia, May 28 - June 1, 2006, Proceedings},
  pages 486--503, 2006.

\bibitem{DBLP:conf/pet/AlhadidiMFD12}
Dima Alhadidi, Noman Mohammed, Benjamin C.~M. Fung, and Mourad Debbabi.
\newblock Secure distributed framework for achieving $\epsilon$-differential
  privacy.
\newblock In {\em Privacy Enhancing Technologies}, pages 120--139, 2012.

\bibitem{damgaard2006unconditionally}
Ivan Damg{\aa}rd, Matthias Fitzi, Eike Kiltz, Jesper~Buus Nielsen, and Tomas
  Toft.
\newblock Unconditionally secure constant-rounds multi-party computation for
  equality, comparison, bits and exponentiation.
\newblock In {\em Theory of Cryptography}, pages 285--304. Springer, 2006.

\bibitem{DBLP:conf/tcc/CramerDI05}
Ronald Cramer, Ivan Damg{\aa}rd, and Yuval Ishai.
\newblock Share conversion, pseudorandom secret-sharing and applications to
  secure computation.
\newblock In {\em Theory of Cryptography, Second Theory of Cryptography
  Conference, {TCC} 2005, Cambridge, MA, USA, February 10-12, 2005,
  Proceedings}, pages 342--362, 2005.

\bibitem{DBLP:conf/focs/McGregorMPRTV10}
Andrew McGregor, Ilya Mironov, Toniann Pitassi, Omer Reingold, Kunal Talwar,
  and Salil~P. Vadhan.
\newblock The limits of two-party differential privacy.
\newblock In {\em FOCS}, pages 81--90, 2010.

\bibitem{DBLP:conf/crypto/MironovPRV09}
Ilya Mironov, Omkant Pandey, Omer Reingold, and Salil~P. Vadhan.
\newblock Computational differential privacy.
\newblock In {\em Advances in Cryptology - {CRYPTO} 2009, 29th Annual
  International Cryptology Conference, Santa Barbara, CA, USA, August 16-20,
  2009. Proceedings}, pages 126--142, 2009.

\bibitem{DBLP:journals/jmlr/ChaudhuriSS13}
Kamalika Chaudhuri, Anand~D. Sarwate, and Kaushik Sinha.
\newblock A near-optimal algorithm for differentially-private principal
  components.
\newblock {\em Journal of Machine Learning Research}, 14(1):2905--2943, 2013.

\bibitem{DBLP:conf/kdd/McSherryM09}
Frank McSherry and Ilya Mironov.
\newblock Differentially private recommender systems: Building privacy into the
  netflix prize contenders.
\newblock In {\em Proceedings of the 15th {ACM} {SIGKDD} International
  Conference on Knowledge Discovery and Data Mining, Paris, France, June 28 -
  July 1, 2009}, pages 627--636, 2009.

\bibitem{DBLP:conf/stoc/DworkTT014}
Cynthia Dwork, Kunal Talwar, Abhradeep Thakurta, and Li~Zhang.
\newblock Analyze gauss: optimal bounds for privacy-preserving principal
  component analysis.
\newblock In {\em Symposium on Theory of Computing, {STOC} 2014, New York, NY,
  USA, May 31 - June 03, 2014}, pages 11--20, 2014.

\bibitem{DBLP:conf/nips/ChaudhuriV13}
Kamalika Chaudhuri and Staal~A. Vinterbo.
\newblock A stability-based validation procedure for differentially private
  machine learning.
\newblock In {\em NIPS}, pages 2652--2660, 2013.

\bibitem{DBLP:conf/nips/ChaudhuriM08}
Kamalika Chaudhuri and Claire Monteleoni.
\newblock Privacy-preserving logistic regression.
\newblock In {\em Advances in Neural Information Processing Systems 21,
  Proceedings of the Twenty-Second Annual Conference on Neural Information
  Processing Systems, Vancouver, British Columbia, Canada, December 8-11,
  2008}, pages 289--296, 2008.

\bibitem{DBLP:conf/icml/0002T13}
Prateek Jain and Abhradeep Thakurta.
\newblock Differentially private learning with kernels.
\newblock In {\em Proceedings of the 30th International Conference on Machine
  Learning, {ICML} 2013, Atlanta, GA, USA, 16-21 June 2013}, pages 118--126,
  2013.

\bibitem{DBLP:journals/pvldb/ZhangZXYW12}
Jun Zhang, Zhenjie Zhang, Xiaokui Xiao, Yin Yang, and Marianne Winslett.
\newblock Functional mechanism: Regression analysis under differential privacy.
\newblock {\em {PVLDB}}, 5(11):1364--1375, 2012.

\bibitem{DBLP:conf/sigmod/ZhangCPSX14}
Jun Zhang, Graham Cormode, Cecilia~M. Procopiuc, Divesh Srivastava, and Xiaokui
  Xiao.
\newblock Privbayes: private data release via bayesian networks.
\newblock In {\em SIGMOD Conference}, pages 1423--1434, 2014.

\end{thebibliography}
\bibliographystyle{unsrt}

\end{document}